\documentclass{article}

\usepackage{amsthm}
\newtheorem{theorem}{Theorem}[section]
\newtheorem{lemma}[theorem]{Lemma}
\newtheorem{proposition}[theorem]{Proposition}
\newtheorem{definition}[theorem]{Definition}
\newtheorem{problem}[theorem]{Problem}
\newtheorem{protocol}[theorem]{Protocol}

 \usepackage{url}
 \usepackage{authblk} 
\usepackage[utf8]{inputenc} 
\usepackage[T1]{fontenc}
\usepackage{amsmath}
\usepackage{mathtools} 

\usepackage{txfonts}

\begin{document}

\title{On key exchange protocol based on Two-side multiplication action}

\author[1]{ Otero Sanchez, Alvaro\thanks{Autor de correspondencia: \texttt{aos073@ual.es}}}
\affil[1]{Department of Mathematics, University of Almería, 04120, Almería, Spain}

\maketitle
\begin{abstract}
We present a cryptanalysis of a two key exchange protocol based on two-side multiplication action. For this purpose, we use the properties of the algebraic structures to obtain a linear system whose solution enable to provide an efficient attack that recovers the shared secret key from publicly exchanged information for any instance of the digital semiring and twisted group ring in polynomial time. 
\end{abstract}
The modern cryptosystems appears with the fundational work of Diffie-Hellman \cite{diffie1976new}, with foundations upon the difficulty of the discrete logarithm problem in cyclic groups, especially over elliptic curves. Nevertheless, advanced in quantum computing threat the secutity of those protocols. In particular, in 1994, Shor introduced a quantum algorithm capable of efficiently solving this problem. This marked the start of the research in post-quantum cryptography, dedicated to constructing cryptographic schemes resilient to quantum attacks.

As a possible solution,  Maze et al. \cite{maze2007} proposed a key exchange protocol using  semigroup. This papers is the beginning for cryptography based on semiring, as in that paper a finite congruence simple semiring is proposed as a possible algebraic framework. These algorithm can be interpreted as a generalization of classical protocols such as the Diffie-Hellman \cite{diffie1976new} and ElGamal \cite{elgamal1985public} protocols, but instead of a cyclic group, they use abelian semigroups. Their framework has inspired various cryptographic developments. For example, Kahrobaei and Koupparis \cite{kahrobaei2016group} explored the use of non-abelian group actions, pushing the original idea into the realm of non-commutative algebra. Similarly, Gnilke and Zumbrägel \cite{gnilke2024cryptographic} connected these concepts with recent progress in isogeny-based cryptography. Another extension can be found in the work of Torrecillas - Olvera - Lopez \cite{Olvera19}, who applied twisted group rings to design new key exchange mechanisms. However, the instance proposed in \cite{maze2007} has been shown recently to be vulnerable to cryptanalysis \cite{otero2024cryptanalysis}. Here, we present a generalization of the last one and show how to apply it to different algebraic structures. 

In another line of research, Grigoriev and Shpilrain investigated the use of tropical semirings as a foundation for public key cryptography, including both key exchange protocols \cite{Grigoriev2013, Grigoriev2019} and digital signature schemes \cite{Chen2024}. Nevertheless, Kotov and Ushakov \cite{Kotov2018} introduced a heuristic algorithm known as the Kotov-Ushakov attack which has become a standard method of cryptanalysis in tropical cryptography. However, that algorithm is based on the solution of an minimal set cover problem, which is one of Karp’s 21 prob-
lems shown to be NP-complete in 1972. Recently, Otero et all \cite{sanchez2024} recently proposed a deterministic alternative that avoids NP-problems and find a solution in polynomial time. Moreover, other schemes based on tropical algebra by Grigoriev and Shpilrain have also been shown to lack security, as demonstrated in several studies \cite{Muanalifah2022, Isaac2021, Rudy2020}.

In \cite{Kahrobaei2013}, the authors propose the use of a group ring to perform a key exchange protocol. This idea has been generalized in other works, such as in \cite{Olvera19}, where the action is modified to a two-sided action similar to that of \cite{maze2007}, and the group ring is twisted by a 2-cocycle with the dihedral group. Another example is presented in \cite{deLaCruz2024}, where the subspaces on which the two-sided action is performed and the twist are altered.

However, the latter approach was cryptanalyzed recently in \cite{Tinani2024}, where the authors reduce the two-sided problem to a system of equations over circulant matrices, using specific equations that arise when the base group is the dihedral group. They provide a probabilistic solution, where the attacker must find an element by random sampling under certain conditions. Nevertheless, the security of this approach in the context of a different group or a new twist remains an open question. In this paper, we introduce a novel approach that allows for a comprehensive cryptanalysis of all such cases.

More recently, new directions have emerged leveraging semiring structures for practical cryptographic applications. In particular, Nassr et al. proposed a public-key encryption scheme grounded on the hardness of the two-sided digital circulant matrix action problem over a semiring initially introduced by Huang et al. \cite{nassr2025}.  

Between the time our preprint was made publicly available and the final publication of our article, a paper was published presenting a cryptanalysis of the protocol based on digital numbers \cite{ponmaheshkumar2025}. However, as the authors of that work acknowledge in their introduction, our results precede the publication of their article. Moreover, the approach by which the solutions are obtained differs significantly from theirs and is both distinct and independent. In addition, we have added the finite ring section to show how our method is different and able to analyze a wide range of cryptosystems

In this paper, we will show how a similar method can be applied to cryptoanalyce different protocols based on two side multiplication, providing as example the protocol over digital semiring of \cite{nassr2025} and other over twisted group ring as in \cite{Olvera19}.

\section{General attack against two side action} \label{sectio2}
Modern example of two side action on cryptogrpahy are base in the original paper \cite{maze2007}, In that paper, the following general setting is presented

\begin{definition}
Let $S$ be a semiring. A left $S$-semimodule is a commutative monoid $(Mo, +, 0_M)$ equipped with a scalar multiplication $S \times Mo \to Mo$, denoted $(s, m) \mapsto sm$, such that for all $s, t \in S$ and all $m, n \in Mo$, the following axioms hold:

\begin{enumerate}
  \item $s(m + n) = sm + sn$ \hfill (left distributivity)
  \item $(s + t)m = sm + tm$ \hfill (right distributivity)
  \item $(st)m = s(tm)$ \hfill (associativity)
  \item $1_S m = m$ \hfill (identity)
  \item $0_S m = s 0_M = 0_M$ \hfill (annihilation)
\end{enumerate}

If $S$ is a semiring with multiplicative identity $1_S$, then $1_S$ acts as the identity on $Mo$.

A right $S$-semimodule is defined analogously, with scalar multiplication $Mo \times S \to Mo$.
\end{definition}

The key exchange protocol prensented is 
\begin{protocol} \label{KeyExchange}
    Let $S$ be a semiring.  Alice and Bob agree  $M \in S$ and two commutative set $T_1,T_2\subset S$.
    \begin{enumerate}
        \item Alice chose $(A_1,A_2)\in T_1 \times T_2$ and makes public $pk_1=A_1 M  A_2$
        \item Bob chose $(B_1,B_2)\in T_1 \times T_2$ and makes public $pk_1=B_1 M  B_2$
        \item Alice computes $A_1 pk_2 A_2$ and Bob computes $B_1 pk_1  B_2$
        \end{enumerate}
        The common key is 
        \[
        A_1  pk_2   A_2 = A_1  B_1   M   B_2   A_2 = B_1   A_1   M   A_2   B_2 = B_1  pk_1   B_2
        \]
\end{protocol}

Those key exchange protocol are so called two-side action key exchange. In some papers, as in \cite{maze2007}, $T_i=C[M_i]=\{ \sum_{j=0}^m r_j M_i^j ; r_j \in Z(S)\}$ with $Z(S)$ the center of the semiring $S$. Other commutative sets are the circulant matrix, as in \cite{nassr2025}. 

In general, the security of such protocol relies on the following problems

\begin{problem}[SAP]
Let $S$ be a semiring. Let \( (A_1, A_2) \in T_1 \times T_2 \) with $T_1,T_2$ commutative sets of $S$, and let \( U = A_1  M  A_2 \) for an arbitrary element \( M \in S \). Given \( U \) and \( M \), the challenge is to obtain two  elements \( (A_1', A_2') \in T_1 \times T_2 \) such that
\[
U = A_1' M A_2'.
\]
\end{problem}

\begin{problem}[Diffie-hellman Problem over semiring]
Let $S$ be a semiring. Let \( (A_1, A_2),(B_1, B_2) \in T_1 \times T_2 \) be elements with $T_1,T_2$ commutative sets of $S$, and \( U = A_1 M  A_2 \), and \( V = B_1 M  B_2 \) for an arbitrary element \( M \in S\). Given \( U \), \( V \), and \( M \), the challenge is to obtain
\[
K = A_1  B_1  M  B_2  A_2.
\]
\end{problem}

\begin{problem}[Decisional problem over semiring]
Let $S$ be a semiring. Given \( M, U \in S \), do they exist two elements \( (A_1, A_2) \in T_1 \times T_2 \) such that
\[
U = A_1  M  A_2
\]
\end{problem}

Now, we will suppose that the sets $T_1, T_2$ have a system of generators $\{L_i^j\}_{i=1}^{m_j} \subset T_j$ over the center of $S$, $j=1,2$. Also, suppose that there exists an algorithm to find a solution to the system $\sum_{i=1}^N a_i H_i = Y$ with $Y,H_i\in S, \forall i = 1,\cdots N$ $a_i \in Z(S) \forall i =1,\cdots N$. Under this situation, we will solve Diffie-hellman Problem over semiring.

Let \( (A_1, A_2),(B_1, B_2) \in T_1 \times T_2 \), and \( U = A_1 M  A_2 \), and \( V = B_1 M  B_2 \) for an arbitrary element \( M \in S \). Then, $A_1,A_2$ can be written as
\begin{equation}
    A_j = \sum_{i=1}^{n_1} c_i^j L_i^j
\end{equation}
for $c_i^j\in Z(S)$ unknown $j=1,2$. Then, we have that
\begin{equation}
    A_1 M A_2 = \left(\sum_{i=1}^{n_1} c_i^1 L_i^1\right) M \left(\sum_{i=1}^{n_2} c_i^2 L_i^2\right) = \sum_{i=1}^{n_1} \sum_{j=1}^{n_2}c_i^1 c_j^2 L_i^1 M L_j^2
\end{equation}

The solution $c_i^1c_j^2$ are particular solution of following system
\begin{equation}
    A_1 M A_2 = \sum_{i=1}^{n_1} \sum_{j=1}^{n_2}z_{ij}L_i^1 M L_j^2
\end{equation}

However, any solution of that system make unsafe the protocol, as in \cite{otero2024cryptanalysis} we can perform the following identity

\begin{align}
   \sum_{i=1}^{n_1} \sum_{j=1}^{n_2}z_{ij}L_i^1 V L_j^2 & = \sum_{i=1}^{n_1} \sum_{j=1}^{n_2}z_{ij}L_i^1 B_1 M L_2 L_j^2 \\ & = \sum_{i=1}^{n_1} \sum_{j=1}^{n_2}z_{ij} B_1L_i^1 M L_j^2  B_2 \\ & = B_1 \left(\sum_{i=1}^{n_1} \sum_{j=1}^{n_2}z_{ij} L_i^1 M L_j^2 \right)  B_2 \\ & = B_1 A_1 M A_2 B_2 
\end{align}
which is the private key.

Therefore, the security of such protocol relies on the dificulty of finding a solution of the system $A_1 M A_2=\sum_{i=1}^{n_1} \sum_{j=1}^{n_2}z_{ij}L_i^1 M L_j^2$

\section{Cryptoanalysis of some protocols}
In this section we will present some example of two-side multiplication key exchange that are not safe due to this approach. 

\subsection{Digital semiring}
We will introduce some basic background on tropical semiring as well as digital semiring

In \cite{huang2024}, a new additively semiring is proposed, which they call the digits semiring.
\begin{definition}
    Let $W=\mathbb{N}\cup \{\infty\}$ and for all $g\in \mathbb{N}$, let $\delta(g)$ be the sum of all digits of $g$. The digits semiring is the semiring $(W,\oplus, \otimes)$ with
    \[
g_1 \oplus g_2 =
\begin{cases}
g_1 & \text{if } \delta(g_2) < \delta(g_1), \\
g_2 & \text{if } \delta(g_2) > \delta(g_1), \\
\max(g_1, g_2) & \text{if } \delta(g_1) = \delta(g_2),
\end{cases}
\]

\[
g_1 \otimes g_2 =
\begin{cases}
g_1 & \text{if } \delta(g_1) < \delta(g_2), \\
g_2 & \text{if } \delta(g_1) > \delta(g_2), \\
\min(g_1, g_2) & \text{if } \delta(g_1) = \delta(g_2),
\end{cases}
\]
\end{definition}
To differenciate the natural order of $\mathbb{N}$ and $W$ induced by addition, we will note $\leq_N$ the natural order of numbers, and $\leq_W$ the one given in $W$. Note that all additively idempotent semiring $R$ have an induced order by $a\leq_R b$ if and only if $a + b = b$.

Over all semiring we can define the semiring of matrix with coefficients in such semiring.

\begin{definition}
Let \( R \) be a semiring. Then the set of squared matrix over $R$, \( Mat_n(R) \), is a semiring with the usual operations
\begin{itemize}
    \item   \( (A\oplus B)_{ij} = A_{ij}\oplus B_{ij} \),
    \item   \(  (A\otimes B)_{ij} = \bigoplus_{k=1}^{n} A_{ik} \otimes B_{kj} \),
\end{itemize}
\end{definition}
Note that if $R$ is additively idempotent, then so is \( Mat_n(R) \)

\begin{definition}
    Let $R$ be a semiring. A matrix $C\in \rm{Mat}_n(R)$ is called circulant if there are $c_0,c_1,\cdots, c_{n-1}\in R$ such that
    \[
C=\begin{pmatrix}
c_0 & c_{n-1} & c_{n-2} & \cdots & c_1 \\
c_1 & c_0 & c_{n-1} & \cdots & c_2 \\
c_2 & c_1 & c_0 & \cdots & c_3 \\
\vdots & \vdots & \vdots & \ddots & \vdots \\
c_{n-1} & c_{n-2} & c_{n-3} & \cdots & c_0
\end{pmatrix}
\]
We will denote $C$ as $C=Circ(c_0,\cdots, c_{n-1})$
\end{definition}

A famous result regarding the structure of circulant matrix is
\begin{theorem}
    The set $Circ_n(R)$ of ciruclar matrix of $n\times n$ over $R$ form a commutative subsemiring of $\rm{Mat}_n(R)$. 
\end{theorem}

In \cite{huang2024}, the following key exchange protocol is proposed

\begin{protocol} \label{KeyExchange}
    Let $W$ be the digital semiring. Alice and Bob agree matrix $M \in \rm{Mat}_n(W)$.
    \begin{enumerate}
        \item Alice chose $A_1,A_2\in Circ_n(\mathbb{W)}$ and makes public $pk_1=A_1 \otimes M \otimes A_2$
        \item Bob chose $B_1,B_2\in Circ_n(\mathbb{W)}$ and makes public $pk_1=B_1 \otimes M \otimes B_2$
        \item Alice computes $A_1\otimes pk_2 \otimes A_2$ and Bob computes $B_1\otimes pk_1 \otimes B_2$
        \end{enumerate}
        The common key is 
        \[
        A_1\otimes pk_2 \otimes A_2 = A_1\otimes B_1 \otimes M \otimes B_2 \otimes A_2 = B_1 \otimes A_1 \otimes M \otimes A_2 \otimes B_2 = B_1\otimes pk_1 \otimes B_2
        \]
\end{protocol}
The security of this protocol is based on the following problem

\begin{problem}[MAP {\cite{huang2024}}]
Let \( A_1, A_2 \in M_n(G) \) be two circulant matrices, and let \( U = A_1 \otimes M \otimes A_2 \) for an arbitrary matrix \( M \in M_n(G) \). Given \( U \) and \( T \), the challenge is to obtain two circulant matrices \( A_1' \) and \( A_2' \) such that
\[
U = A_1 \otimes M \otimes A_2.
\]
\end{problem}

In \cite{huang2024}, it is shown that MAP can be transformed into the problem of solving quadratic polynomial systems on the semiring \( (W, \oplus,  \otimes) \), which is proven to be an NP-problem

\begin{proposition}[\cite{huang2024}]
MAP can be transformed to the problem of solving quadratic polynomial systems on the digits semiring.
\end{proposition}

Finally, we introduce the concept of maximal solution,
\begin{definition}
Let $R$ be an additively idempotent semiring, and let $XA=Y$ be a linear system of equations. We say that $\hat{X}$ is the maximal solution of the system if and only if the two following conditions are satisfied 
    \begin{enumerate}
        \item $\hat{X}\in R^n$ is a solution of the system, i.e. $\hat{X}A=Y$,
        \item if $Z\in R^n$  is any other solution of the system, then $Z+\hat{X}=\hat{X}$.
    \end{enumerate}
\noindent This last condition is equivalent to $Z\leq \hat{X}$.
    
\end{definition}

In \cite{sanchez2024} a new method to solve linear equations over additively idempotent semiring is proposed, as well as its cryptographic applications. In \cite{nassr2025} they assert that the cryptoanalysis on such paper can not be used against \ref{KeyExchange}, as the private keys do not come from tropical polynomials of matrices. We will present a modification of that cryptanalysis that can be used with circulant matrix.

First, let $C_i = C[e_i]$ with $e_i$ the $i-$th vector of the canonical base. Then, we have that
\[C[a_1,a_2,\cdots, a_n]= a_1 C_1 \oplus a_2C_2\cdots \oplus a_n C_n\]
and therefore they form a commutative basis. As a result, we have to solve 

\begin{equation}
    pk_1= \bigoplus_{i,j=1}^n z_{ij} C_i \otimes M \otimes C_j
\end{equation}

To solve the previous linear system, we must note that
\begin{lemma}
    Let $a,b\in W$, then $a\otimes b \leq  a,b$, where order is based the natural order of additively idempotent semiring.
\end{lemma}
\begin{proof}
    w.l.o.g we can assume $a\leq_W b$. We have that $a\leq_W b$ means that $\delta(a)\leq_N \delta(b)$ or $\delta(a) = \delta(b)$  and $a \leq_{N} b$. In both cases, $a\otimes b = a \leq _W b $. 
\end{proof}

In \cite{sanchez2024} the following characterization of maximal solution is presented.

\begin{theorem} \label{SolucionSistemaGeneral}
Given $(R,+,\cdot)$ an additively idempotent semiring, let $T_i = \{ x \in R : x H_i + Y = Y \}$ $\forall i =1,\dots,n$. Suppose that these subsets have a maximum with respect to the order induced in $R$

    \begin{equation*}
        C_i = \max T_i.
    \end{equation*}

\noindent If $XH = Y$ has a solution, then $Z=(C_1,\dots C_n)$ is the maximal solution of the system.
\end{theorem}
Note that if $\delta(a) \leq \delta(y)$, then $x\otimes a \leq_W y$ for all $x\in W$, and $\delta(a) > \delta(y)$ then $\delta(x) \leq \delta(y)$. As a result, $T_i = W$ if $H_i \leq_W \{ y_j ; \delta(y_j) \leq \delta (h_{ij})\}$, and $x\leq_W \min_W Y$ in other case, with $\min_W$ the min respect to the order in $W$. Therefore, 
    \[
\max T_i =
\begin{cases}
\infty & \text{if } H_i \leq_W Y,  \\
\min_W Y & \text{if } other
\end{cases}
\]
As there must be a solution of the system, provided that $K=A_1MA_2$, so we can compute a solution $d_{ij}$ of $K= \oplus_{i,j=1}^n z_{ij} C_i \otimes M \otimes C_j$.

\subsection{Finite ring}
There are some key exchange protocol where the semiring proposed is indeed a ring, as in the case of \cite{Olvera19}, \cite{deLaCruz2024}. The last one was cryptoanalyced by \cite{Tinani2024}, but only with the use of the specific relations of the group and the commutative sets, that yield in an algorithm that is not replicable in other cases. In fact, the security of the general case was still an open problem. 

\begin{definition}
    Let $G$ a non abelian semigroup, $T\subset G$ a subset and let $a\in T$. The adjoint is a map $(-)^*T:\longrightarrow G$ such that
    \begin{equation}
        a \cdot b^*  = b\cdot a^*  \text{  } \forall a,b \in T
    \end{equation}
\end{definition}

In \cite{Olvera19}, the authors introduce the following generalization of the clasical diffie-hellman based on group action, 

\begin{protocol}
    Let $S$ be a finite set, $G$ be a non-abelian semigroup, and $\varphi$ a $G$-action on $S$, and a public element $h \in S$. The extended Diffie–Hellman key exchange in $(G, S, \varphi)$ is the following protocol:

\begin{enumerate}
  \item Alice chooses $a \in G$ and computes $\varphi(a, h)$. Alice’s private key is $a$, and her public key is $pk_A = \varphi(a, h)$.
  \item Bob chooses $b \in G$ and computes $\varphi(b, h)$. Bob’s private key is $b$, and his public key is $pk_B = \varphi(b, h)$.
  \item Their common secret key is then
  \[
    \varphi(a^*, pk_B) = \varphi(a^*, \varphi(b, h)) = \varphi(a^*b, h) = \varphi(b^* a, h) = \varphi(b^*, \varphi(a, h)) = \varphi(b^*, pk_A),
  \]
\end{enumerate}
\end{protocol}

As a semigroup, they will use the multiplicative semigroup of a twised group ring. To present this algebraic structure, we recall the definition of $2-$cocycle

\begin{definition}
Let $G$ be a group and $A$ be an abelian group. An application
\[
\alpha : G \times G \rightarrow A
\]
is a \textbf{2-cocycle} if:

\begin{enumerate}
  \item $\alpha(g, 1) = \alpha(1, g) = 1$, for all $g \in G$,
  \item $\alpha(g, h)\alpha(gh, k) = \alpha(g, hk)\alpha(h, k)$, for all $g, h, k \in G$.
\end{enumerate}
\end{definition}

\begin{definition}
    Let $K$ be a ring, $G$ a group and $\alpha : G \times G \longrightarrow U(K)$ a $2$-cocycle to the units of $K$. Then the twisted group ring $K^\alpha G$ is the set of $K-$vector space spaned by $G$ with multiplication 
    \begin{equation}
        (a\overline{g} )( b \overline{h}) = ab\alpha(g,h) \overline{gh} \text{ } \forall a,b \in K, \overline{g}, \overline{h} \in G
    \end{equation}
    and expanded by linearity.
\end{definition}

We need the previous lemma.
\begin{lemma}
    Let $K$ be a commutative ring, let $T_1 = K C$ with $C\subset Z(G)$ and $\alpha(g,h)=\alpha(h,g)  \forall g,h\in C$. Then $T_1$ is a commutative ring
\end{lemma}
\begin{proof}
It is enough to prove the commutativity for $a\overline{g}, b\overline{h} \in T_1$
\begin{equation}
    (a\overline{g})  (b\overline{h}) = ab \alpha(g,h) \overline{gh} = ba \alpha(h,g) \overline{hg} = (b\overline{h})(a\overline{g})
    \end{equation}
\end{proof}

Now, let $T_1 = K C$ as in the previous lemma, and let $T_2=K H$ with $H \subset G$ such that and adjoint $(-)^* : T_2 \longrightarrow K^\alpha G$ exist. Then the key exchange protocol will be 

\begin{protocol}
    Let $K^\alpha G$ be twisted group ring, and $h\in K^\alpha G$

\begin{enumerate}
  \item Alice chooses $(a,b) \in T_1 \times T_2$ and computes $ahb$. Alice’s private key is $(a,b)$, and her public key is $pk_A = ahb$.
  \item Bob  chooses $(c,d) \in T_1 \times T_2$ and computes $chd$. Bob’s private key is $(c,d)$, and her public key is $pk_B = chd$.
  \item Alice compute $ap_B b^*$ and Bob $cp_A d^*$ 

Their common secret key is then
  \[
    apk_B b^*= a c h d b^* = ca h bd^* = c pk_A d^*
  \]
\end{enumerate}
\end{protocol}

To start the cryptanalysis, recall that all finite ring $R$ is a $\mathbb{Z}_p$- vector space for some $p | ch (R)$. 

Let $\{v_1, \cdots, v_n\}$ a base of $K$ as a $\mathbb{Z_p}$-vector space. Then $\{v_ig ; i=1,\cdots , n, g \in C\}, \{v_ig ; i=1,\cdots , n, g \in H\}$ are the bases of $T_1$, $T_2$, respectively. To apply the general attack defined in previous sections, we must take into account the adjoint. For this, if we find a solution 

\begin{equation}
    ahb = \sum_{i=1}^{n_1} \sum_{j=1}^{n_2}z_{ij}L_i^1 h L_j^2
\end{equation}
with $\{L_ij\}_j$ a basis of $T_i$ for $i=1,2$, then

\begin{align*}
    \sum_{i=1}^{n_1} \sum_{j=1}^{n_2}z_{ij}L_i^1 pk_B (L_j^2)^* & = 
\sum_{i=1}^{n_1} \sum_{j=1}^{n_2}z_{ij}L_i^1 chd (L_j^2)^* \\&= 
\sum_{i=1}^{n_1} \sum_{j=1}^{n_2}z_{ij} c L_i^1 h L_j^2 d^* \\&= 
c \left(\sum_{i=1}^{n_1} \sum_{j=1}^{n_2}z_{ij} L_i1 h L_j2 \right)  d \\&=  c a h b d^*
\end{align*}
which is the private key.

To ilustrate the previous algorithm, in the case proposed in \cite{Olvera19}, they propose a finite field $K$, its primitive root of unity $t$, and the dihedral group of $2m$ elements,
\[
D_{2m} = \langle x, y \mid x^m = y^2 = 1,\ yx^a = x^{m-a}y \rangle.
\]
Then $R = K^{\alpha} D_{2m}$, where $\alpha$ is the $2$-cocycle
\[
\alpha : D_{2m} \times D_{2m} \rightarrow K
\]
defined by:
\[
\alpha(x^i, x^j y^k) = 1,\quad \alpha(x^i y, x^j y^k) = t^{j}, \quad \text{for } i, j = 1, \ldots, 2m - 1,
\]
is a twisted group ring.

\begin{definition}
    Let $R = K^{\alpha} D_{2m}$, where $t$ is a primitive root of unity that generates $K$ and $\alpha$ is the $2$-cocycle defined above. Given $h \in R$,

\[
h = \sum_{\substack{0 \leq i \leq m-1 \\ k=0,1}} r_i x^i y^k,
\]

where $r_i \in K$ and $x, y \in D_{2m}$. Then we define $h^* \in K^{\alpha} D_{2m}$ as:

\[
h^* = \sum_{\substack{0 \leq i \leq m-1 \\ k=0,1}} r_i t^{-i} x^i y^k,
\]

where $r_i \in K$ and $x, y \in D_m$.

\end{definition}

If we denote $C_m$ as the cyclic subgroup of $D_{2m}$ of order $m$, then $R = K^\alpha D_{2m}$ can be written as
\[
R = R_1 \oplus R_2,
\]
where $R_1 = K C_m$ and $R_2 = K^\alpha C_m y$. In this context, they can define $A_j \leq R_j$ as
\[
A_j = \left\{ \sum_{i=0}^{m-1} r_i x^i y^k \in R_j : r_i = r_{m - i} \right\}.
\]

\begin{proposition}
Given $h_1, h_2 \in R$,
\begin{itemize}
    \item If $h_1, h_2 \in R_1$, then $h_1 h_2 = h_2 h_1$;
    \item If $h_1, h_2 \in A_2$, then $h_1 h_2^* = h_2 h_1^*$, and $h_1^* h_2 = h_2^* h_1$;
    \item If $h_1 \in A_1$, $h_2 \in A_2$, then $h_1 h_2 = h_2 h_1^*$.
\end{itemize}
\end{proposition} 

Then the key exchange protocol will be:

Let $h \in R$ be a random public element. The key exchange between Alice and Bob proceeds as follows:

\begin{enumerate}
    \item Alice selects a secret pair $s_A = (g_1, k_1)$, where $g_1 \in R_1$, $k_1 \in A_2 \leq R_2$.
    
    \item Bob selects a secret pair $s_B = (g_2, k_2)$, where $g_2 \in R_1$, $k_2 \in A_2 \leq R_2$.
    
    \item Alice sends Bob the element $p_A = g_1 h k_1$, and Bob sends Alice $p_B = g_2 h k_2$.
    
    \item Alice computes the shared key
    \[
    K_A = g_1 p_B k_1^*,
    \]
    and Bob computes
    \[
    K_B = g_2 p_A k_2^*.
    \]
    
    \item Then $K_A = K_B$, and both parties share the same secret key.
\end{enumerate}

Under this situation, we have that $K=\mathbb{F}_{p^n}$, and that $R_1$ is a $\mathbb{F}_p$-vector space with commutiative basis 
\begin{equation}
    \{t^ix^j; i=0,\cdots, n-1, j=0,\cdots, m-1\}
\end{equation}
 and that $A_2$ is a $\mathbb{F}_p$-vector space with basis 
\begin{equation}
    \{t^i(x^j+x^{m-j}) ; i=0,\cdots, n-1, j=1,\cdots, \left\lfloor \frac{m-1} {2}\right\rfloor\} \cup \{ t^i ; i=0,\cdots, n-1,\}
\end{equation}
if $m$ is even, and 
\begin{equation}
    \{t^i(x^j+x^{m-j}) ; i=0,\cdots, n-1, j=1,\cdots, \left\lfloor \frac{m-2} {2}\right\rfloor\} \cup \{ t^i, t^ix^{m/2}; i=0,\cdots, n-1,\}
\end{equation}
if $m$ is odd. Therefore we can compute a linear system as in section \ref{sectio2} over the field $\mathbb{F}_p$, from which it is possible to obtain the key

\section{Conclusions}
We have cryptoanalyced some key exchange protocol based on two-side multiplication action. We have use this algorithm to obtain the share key in  the public key exchange proposed in \cite{huang2024} and \cite{Olvera19}. For the first one we have used the original ideas of \cite{sanchez2024} to the special case of digital sum, find a method to obtain the maximal solution of a linear system over such semiring, and for the last one we have use the properties of finite field to obtain a linear system over a finite field. 

\section*{Acknowledgments}
This research was supported by the Spanish Ministry of Science, Innovation and Universities under the FPU 2023 grant program.

\section*{ORCID}

\noindent Alvaro Otero Sanchez - \url{https://orcid.org/0009-0009-5613-2081}

\end{document}